\documentclass[11pt]{article}

\pdfoutput=1
\usepackage[utf8]{inputenc}
\usepackage{amsmath,amsfonts,hyperref,amsthm, graphicx, comment}
\usepackage[textsize=footnotesize,color=green!40]{todonotes}
\usepackage{fullpage,bbm}
\usepackage{authblk}
\usepackage[inline]{enumitem}
\usepackage{thm-restate}
\hypersetup{
	colorlinks,
	linkcolor={blue!80!black},
	citecolor={blue!80!black},
	urlcolor={blue!80!black}
}

\newtheorem{theorem}{Theorem}
\newtheorem{lemma}[theorem]{Lemma}

\theoremstyle{definition}
\newtheorem*{definition*}{Definition}
\newtheorem{remark}{Remark}
\newtheorem{claim}[theorem]{Claim}
\newtheorem{conjecture}[theorem]{Conjecture}
\newtheorem{question}[theorem]{Question}

\newcommand{\cG}{\mathcal{G}}

\title{A polynomial version of Cereceda's conjecture\footnote{This work was supported by ANR project GrR (ANR-18-CE40-0032).}}
\author[1]{Nicolas Bousquet} 
\author[2]{Marc Heinrich} 

\affil[1]{Univ. Grenoble Alpes, CNRS, Laboratoire G-SCOP, Grenoble-INP, Grenoble, France.\thanks{nicolas.bousquet@grenoble-inp.fr}}
\affil[2]{LIRIS, Universit\'e Claude Bernard, Lyon, France}

\date{}

\begin{document}

\maketitle

\begin{abstract}
Let $k$ and $d$ be such that $k \ge d+2$. Consider two $k$-colourings of a $d$-degenerate graph $G$. Can we transform one into the other by recolouring one vertex at each step while maintaining a proper coloring at any step? Cereceda et al. answered that question in the affirmative, and exhibited a recolouring sequence of exponential length. However, Cereceda conjectured that there should exist one of quadratic length.

The $k$-reconfiguration graph of $G$ is the graph whose vertices are the proper $k$-colourings of $G$, with an edge between two colourings if they differ on exactly one vertex. Cereceda's conjecture can be reformulated as follows: the diameter of the $(d+2)$-reconfiguration graph of any $d$-degenerate graph on $n$ vertices is $O(n^2)$. So far, the existence of a polynomial diameter is open even for $d=2$.

In this paper, we prove that the diameter of the $k$-reconfiguration graph of a $d$-degenerate graph is $O(n^{d+1})$ for $k \ge d+2$. Moreover, we prove that if $k \ge \frac 32 (d+1)$ then the diameter of the $k$-reconfiguration graph is quadratic, improving the previous bound of $k \ge 2d+1$.
We also show that the $5$-reconfiguration graph of planar bipartite graphs has quadratic diameter, confirming Cereceda's conjecture for this class of graphs.
\end{abstract}

\section{Introduction}

Reconfiguration problems consist in finding step-by-step transformations between two feasible solutions of a problem such that all intermediate states are also feasible. Such problems model dynamic situations where a given solution already in place has to be modified for a more desirable one while maintaining some properties all along the transformation. Reconfiguration problems have been studied in various fields such as discrete geometry~\cite{BoseLPV18}, optimization~\cite{BBRM18} or statistical physics~\cite{mohar4}.
For a complete overview of the reconfiguration field, the reader is referred to the two recent surveys on the topic~\cite{Nishimura17,Heuvel13}. In this paper, our reference problem is graph colouring.

Let $\Pi$ be a problem and $I$ be an instance of $\Pi$. The \emph{reconfiguration graph} is the graph where vertices are solutions of $I$ and where there is an edge between two vertices if one can transform the first solution into the other in one step (for graph colouring one step means modifying the colour of a single vertex\footnote{Note that recolouring operations have been studied, for instance recolouring using Kempe chains (see e.g.~\cite{BonamyBFJ19}). In this article, we focus on single vertex recolourings.}).
Given a reconfiguration problem, several questions may arise. (i) Is it possible to transform any solution into any other, i.e. is the reconfiguration graph connected? (ii) If yes, how many steps are needed to perform this transformation, i.e. what is the diameter of the reconfiguration graph? 
In this work, we will focus on the diameter of the reconfiguration graph. The diameter of the reconfiguration graph plays an important role, for instance in random sampling, since it provides a lower bound on the mixing time of the underlying Markov chain (and the connectivity of the reconfiguration graph ensures the ergodicity of the Markov chain\footnote{Actually, it only gives the irreducibility of the chain. To get the ergodicity, we also need the chain to be aperiodic. For the chains associated to proper graph colourings, this property is usually straightforward.}). Since proper colourings correspond to states of the antiferromagnetic Potts model at zero temperature, Markov chains related to graph colourings received a considerable attention in statistical physics and many questions related to the ergodicity or the mixing time of these chains remain widely open (see e.g.~\cite{ChenDMPP19,frieze2007survey}).

\paragraph{Graph recolouring.}
All along the paper $G=(V,E)$ denotes a graph, $n$ is the size of $V$ and $k$ is an integer. For standard definitions and notations on graphs, we refer the reader to~\cite{Diestel}.
A \emph{(proper) $k$-colouring} of $G$ is a function $f : V(G) \rightarrow \{ 1,\ldots,k \}$ such that, for every edge $xy\in E$, we have $f(x)\neq f(y)$. Throughout the paper we will only consider proper colourings and will then omit the proper for brevity. The \emph{chromatic number} $\chi(G)$ of a graph $G$ is the smallest $k$ such that $G$ admits a $k$-colouring.
Two $k$-colourings are \emph{adjacent} if they differ on exactly one vertex. The \emph{$k$-reconfiguration graph of $G$}, denoted by $\cG(G,k)$ and defined for any $k\geq \chi(G)$, is the graph whose vertices are $k$-colourings of $G$, with the adjacency relation defined above. 
Cereceda, van den Heuvel and Johnson provided an algorithm to decide whether, given two $3$-colourings, one can be transformed into the other in polynomial time, and characterized graphs for which $\cG(G,3)$ is connected~\cite{Cereceda09,CerecedaHJ11}.
Given any two $k$-colourings of $G$, it is $\mathbf{PSPACE}$-complete to decide whether one can be transformed into the other for $k \geq 4$~\cite{BonsmaC07}. 

The \emph{$k$-recolouring diameter} of a graph $G$ is the diameter of $\cG(G,k)$ if $\cG(G,k)$ is connected and is equal to $+\infty$ otherwise. In other words, it is the minimum~$D$ for which any $k$-colouring can be transformed into any other one through a sequence of at most~$D$ adjacent $k$-colourings.
Bonsma and Cereceda~\cite{BonsmaC07} proved that there exists a family $\mathcal{G}$ of graphs and an integer $k$ such that, for every graph $G \in \mathcal{G}$ there exist two $k$-colourings whose distance in the $k$-reconfiguration graph is finite and super-polynomial in $n$.


Cereceda conjectured that the situation is different for degenerate graphs.
A graph $G$ is \emph{$d$-degenerate} if any subgraph of $G$ admits a vertex of degree at most~$d$. In other words, there exists an ordering $v_1,\ldots,v_n$ of the vertices such that for every $i \le n$, the vertex $v_i$ has at most $d$ neighbours in $v_{i+1},\ldots,v_n$.
It was shown independently by Dyer et al.~\cite{dyer2006randomly} and by Cereceda et al.~\cite{Cereceda09} that for any $d$-degenerate graph $G$ and every $k \geq d+2$, $\cG(G,k)$ is connected. However the (upper) bound on the $k$-recolouring diameter given by these constructive proofs is of order $c^{n}$ (where $c$ is a constant). 
Cereceda~\cite{Cereceda} conjectured that the the diameter of $\cG(G,k)$ is of order $\mathcal{O}(n^2)$ as long as $k \ge d+2$. If correct, the quadratic function is sharp, even for paths or chordal graphs as proved in~\cite{BonamyJ12}.
The existence of a polynomial upper bound instead of a quadratic one also is open, even for $d=2$ or restricted to particular graph classes such as planar graphs. In what follows, we will call respectively the polynomial (resp. quadratic) Cereceda's conjecture the question of proving that the $k$-recolouring diameter of $d$-degenerate graphs is polynomial (resp. quadratic) for $k \geq d+2$.

The quadratic Cereceda's conjecture is known to be true only for $d=1$ (for trees)~\cite{BonamyJ12} and $d=2$ with $\Delta \le 3$~\cite{FeghaliJP16} (where $\Delta$ denotes the maximum degree of the graph). Cereceda~\cite{Cereceda} showed that for any $d$-degenerate graph $G$ and every $k \geq 2d+1$, the diameter of $\cG(G,k)$ is $O(n^2)$. Bousquet and Perarnau proved that when $k \ge 2d+2$, the diameter of $\cG(G,k)$ is linear~\cite{BousquetP16}.

Even if the conjecture is open for general graphs, it has been proved for several graph classes such as chordal graphs~\cite{BonamyJ12} and bounded treewidth graphs~\cite{BonamyB18}. The polynomial Cereceda's conjecture holds if we replace degeneracy by maximum average degree~\cite{BousquetP16,Feghali19}. In particular, it implies that the $8$-recolouring diameter of a planar graph is polynomial. The polynomial Cereceda's conjecture asks for more since it states that the diameter should be polynomial when $k=7$. So far, the best known upper bound on the $7$-recolouring diameter of planar graphs was subexponential~\cite{EibenF18}.

\paragraph{Our result.}
The main result of the paper is the following:

\begin{theorem}
\label{thm:main}
	Let $d,k \in \mathbb{N}$ and $G$ be a $d$-degenerate graph. Then $\cG(G, k)$ has diameter at most:
	\begin{itemize}
    	\item $C n^2$ if $k \geq \frac {3} {2} (d+1)$ (where $C$ is a constant independent of $k$ and $d$), 
        \item $C_\varepsilon n^{\lceil 1/\varepsilon \rceil}$ if $k \geq (1+\varepsilon)(d+2)$ and $0 < \varepsilon < 1$ (where $C_\varepsilon$ is a constant independent from $k$ and $d$), 
        \item $(C n)^{d+1}$ for any $d$ and $k \ge d+2$ (where $C$ is a constant independent from $k$ and $d$).
	\end{itemize}
\end{theorem}

\begin{figure}
\centering
\def\arraystretch{1.2}
\begin{tabular}{|l|c|c|}
  \hline
  $k$-recoloring diameter &  Lower bound & Upper bound \\
  \hline
  $k=d+2$ & $\mathcal{O}(n^2)$ ~\cite{BonamyJ12} & $\mathcal{O}(n^{d+1})$ [Theorem~\ref{thm:main}]\\
  $k = d+3$ & $\mathcal{O}(n)$ & $\mathcal{O}(n^{\left\lceil\frac{d+1}{2}\right\rceil})$ [Theorem~\ref{thm:main}] \\
  $k \ge (1+\varepsilon)(d+1)$ & $\mathcal{O}(n)$ & $\mathcal{O}(n^{\lceil 1/\varepsilon \rceil})$ [Theorem~\ref{thm:main}] \\
  $k \ge \frac 32(d+1)$ & $\mathcal{O}(n)$ & $\mathcal{O}(n^2)$ [Theorem~\ref{thm:main}] \\
  $k \ge  2(d+1)$ & $\mathcal{O}(n)$ & $\mathcal{O}(n)$ ~\cite{BousquetP16}\\
  \hline
\end{tabular}
\caption{Known results on the $k$-recolouring diameter of $d$-degenerate graphs.}
\label{fig:summary}
\end{figure}

In particular, it implies that the $7$-recolouring diameter of planar graphs is polynomial (of order $O(n^6)$), answering a question of~\cite{BousquetP16,Feghali19}, and is quadratic if $k \geq 9$ (improving the recent result of Feghali giving $k \ge 10$~\cite{Feghali19_2}). For general graphs, our result guarantees moreover that the diameter becomes a polynomial independent of $d$ as long as $k \ge (1+\varepsilon) (d+2)$. We also obtain a quadratic diameter when the number of colours is at least $\frac{3}{2} \cdot (d+1)$, improving the result of Cereceda~\cite{Cereceda} who obtained a similar result for $k\ge 2d+1$. Note moreover that Theorem~\ref{thm:main} ensures that the diameter is polynomial as long as~$d$ is a fixed constant, which was open even for $d=2$ (we get a diameter of order $\mathcal{O}(n^3)$ for $d=2$). The main known results on lower and upper bounds on the $k$-recolouring diameter are summarized in Figure~\ref{fig:summary}.

In order to show Theorem~\ref{thm:main}, we need to prove a more general result that also holds for list-colourings. Indeed, we often need to consider induced subgraphs of our initial graph where the colours of some vertices are``frozen'' (i.e. do not change). By considering the list colouring version, we can delete these vertices and remove their colours from the list of all their neighbours. This more general statement implying Theorem~\ref{thm:main} is given in Section~\ref{sec:main}.

We complete Theorem~\ref{thm:main} by proving the quadratic Cereceda's conjecture for planar bipartite graphs. Euler's formula ensures that planar bipartite graphs are $3$-degenerate. We prove the following in Section~\ref{sec:bip}.

\begin{restatable}{theorem}{bipPlanar}
\label{thm:planarBip}
Let $G$ be a planar bipartite graph. The diameter of $\cG(G, 5)$ is at most $4n^2$.
\end{restatable}

Our proof is based on a discharging argument. It is, as far as we know, the first time such a method is used for reconfiguration.

The proofs of both Theorem~\ref{thm:main} and Theorem~\ref{thm:planarBip} consist in showing that, given two colourings $\alpha$ and $\beta$ of $G$, there exists a transformation from $\alpha$ to $\beta$ of the corresponding length. Since our proofs are algorithmic, they also provide polynomial time algorithms that, given two colourings $\alpha$ and $\beta$, output a transformation from $\alpha$ to $\beta$ of length at most the diameter of the reconfiguration graph.

\paragraph{Further work and open problems.} Even if we obtain a polynomial bound on the diameter of the reconfiguration graph, the quadratic conjecture of Cereceda is still open, even for simple classes of graphs such as (triangle-free) planar graphs or $2$-degenerate graphs.

\begin{conjecture}[Cereceda~\cite{Cereceda}]
For every $d$, there exists $C_d$ such that for every $d$-degenerate graph~$G$, the diameter of $\cG(G,k)$ is at most $C_d \cdot n^2$ as long as $k \ge d+2$.
\end{conjecture}

The question of Cereceda in~\cite{Cereceda} is the following ``\textit{For a graph $G$ with $n$ vertices and $k \ge d + 2$ where $d$ is the degeneracy of $G$, the diameter of $\cG(G,k)$ is $O(n^2)$}''. We were not be able to determine if the coefficient in front of $n^2$ has to be a function of $d$ or not. As a consequence we decided to state the weaker possible version of the conjecture. Note that the constants of Theorem~\ref{thm:planarBip}, for chordal graphs~\cite{BonamyJ12}, and for bounded treewidth graphs~\cite{BonamyB18} do not depend on $d$. A stronger conjecture would then be the following:

\begin{question}[Stronger Conjecture]
There exists a constant $C$ such that, for every $d$ and every $d$-degenerate graph $G$ the diameter of $\cG(G,k)$ is at most $C \cdot n^2$ as long as $k \ge d+2$.
\end{question}

Another interesting question is the following: what is the evolution of the diameter when the number of colours increases. With our proof technique, the power in the exponent decreases little by little and finally becomes quadratic when the number of colours is at least $\frac{3}{2} \cdot (d+1)$. Improving the factor $\frac 32$ may be another interesting way of tackling the general conjecture.

We also know that the diameter becomes linear when $k \ge 2d+2$~\cite{BousquetP16}. Is it possible to improve this result? In particular, can we improve the $k$-recolouring diameter for $\frac 32 (d+1) \le k \le 2d+1$?

\begin{question}
Is there $\alpha <2$ and $\beta \in \mathbb{N}$ such that, for every $d$ and every $d$-degenerate graph $G$ the diameter of $\cG(G,k)$ is at most $C_d \cdot n$ when $k \ge \alpha d + \beta$. 
\end{question}

It was also proved in~\cite{BonamyB18} that the diameter becomes linear if $k$ is at least the grundy number plus~$1$, which in particular implies a linear diameter when $k \ge \Delta+2$. When $k=\Delta+1$, Feghali, Jonhson and Paulusma~\cite{FeghaliJP16} proved that the $k$-recolouring graph is composed of isolated vertices plus a unique component of diameter at most $O(n^2)$ (and this non-isolated component of $\cG(G,\Delta+1)$ is exponentially larger than its number of isolated vertices~\cite{BonamyBP18}).
 
One can also notice in Figure~\ref{fig:summary} that, as long as $k \ge d+3$, no non-trivial lower bound is known on the diameter of the reconfiguration graph. The subtle quadratic lower bound of~\cite{BonamyJ12} when $k = d+2$ seems to be hard to adapt if $k \ge d+3$. Developing new techniques to find lower bounds on the $k$-recolouring diameter of a graph is a challenging open problem.

\section{Polynomial Cereceda's conjecture}\label{sec:main}

Let $G$ be a $d$-degenerate graph and let $v_1,\ldots, v_n$ be a degeneracy ordering. We denote by $d^+(v)$ the out-neighbours of $v$ (i.e., neighbours of $v$ which appear later in the ordering). Recall that a graph is $d$-degenerate if there is a degeneracy ordering such that $d^+(v) \leq d$ for every vertex $v$ of the graph. 

Let us first briefly discuss the main ideas of the proof before stating formally all the results. The main idea is to proceed recursively on the degeneracy. More precisely, we want to delete a subset of vertices in order to decrease by one both the degeneracy and the number of colours. In order to do this, observe that if at some point there is one colour $c$ such that every vertex of the graph is either coloured $c$ or has an out-neighbour coloured $c$, then by removing all the vertices coloured $c$, we decrease the number of available colours by $1$, but we also decrease the degeneracy of the graph by $1$. If $H$ is the resulting graph, by applying induction, we can recolour $H$ however we want, and for example, we can remove completely one colour which we can then use to make the two colourings agree on a subset of vertices. If the colour $c$ satisfies the condition above, we will say that the colour $c$ is \emph{full}. Our main objective will consist in finding a transformation from any colouring $\alpha$ of $G$ to some colouring $\alpha'$ of $G$ which has a full colour (Note that any graph can have such a colouring, for example by applying the \textsc{First-Fit} algorithm in the reverse order of the elimination ordering). We will build the colouring $\alpha'$ (and the transformation) incrementally. However in order to do this, we will need to generalise the problem to list colouring. 

A \emph{list assignment} $L$ is a function which associates a list of colours to every vertex $v$. An \emph{$L$-colouring} is a (proper) colouring $\alpha$ of $G$ such that for every vertex $v$, $\alpha(v) \in L(v)$. The total number of colours used by the assignment is $k = |\bigcup_{v \in G} L(v)|$. A list assignment $L$ is \emph{$a$-feasible} if $|L(v)| \geq |d^+(v)| + a + 1$ for every vertex $v \in G$. We just say that it is \emph{feasible} if it is $1$-feasible. We denote by $\cG(G, L)$ the reconfiguration graph of the $L$-colourings of $G$. (One can easily prove by induction, that if a list assignment is $a$-feasible for $a \ge 1$, then $\cG(G,L)$ is connected). We will prove a generalisation of Theorem~\ref{thm:main} in the case of list colourings. Namely, we will prove that:

\begin{theorem}\label{thm:mainlist}
	Let $G$ be a graph and $a \in \mathbb{N}$. Let $L$ be an $a$-feasible list assignment and $k$ be the total number of colours. Then $\cG(G,L)$ has diameter at most:
	\begin{itemize}
	    \item $kn$ if $k \le 2a$.
    	\item $C n^2$ if $k \le 3a$ ($C$ a constant independent of $k,a$), 
        \item $C_\varepsilon n^{\lceil 1/\varepsilon \rceil}$ if $k \leq (1+\frac{1}{\varepsilon}) a$ where $\varepsilon$ is a constant and $C_\varepsilon$ is independent of $k,a$, 
        \item $(Cn)^{k-1}$ , if $a \ge 1$.
	\end{itemize}
\end{theorem}

The proof of Theorem~\ref{thm:main} follows easily from this result. The only point that is not immediate is that the last point of Theorem~\ref{thm:mainlist} implies the last point of Theorem~\ref{thm:main}. In this case, we need a small trick to guarantee that the diameter does not increase if the number of colours increases. Note that the first point of Theorem~\ref{thm:mainlist} implies in terms of classical colouring that the $k$-recolouring diameter is linear when $k \ge 2d+2$, which is an already known result~\cite{BousquetP16}.

\begin{proof}[Proof of Theorem~\ref{thm:main}]
Note that given a $d$-degenerate graphs and $k$ colours, we can consider the list assignment $L$ where $L(v) = [k]$ for every vertex $v$. This list assignment is $a$-feasible with $a = k-d-1$. 

We start with the second point of Theorem~\ref{thm:main}. If $k \ge (1+\varepsilon) (d+1)$ then:
\begin{align*}
\frac k a = \frac {k} {k-d-1} = 1 + \frac {d+1} {k-d-1} \leq 1 + \frac 1 \varepsilon \;.
\end{align*}
By applying the third case of Theorem~\ref{thm:mainlist}, the result follows.

The first point follows immediately from the result above by taking $\varepsilon = \frac 1 2$.

Finally, in order to prove the last point, we need to prove that we can ``replace'' $k$ by $d$.
Let $G$ be a $d$-degenerate graph and let $\gamma$ be a $(d+1)$-colouring of it. Let us prove that, if $k > d+2$, any colouring $\alpha$ can be transformed into $\gamma$ within $O(n^{d+1})$ steps (and not $O(n^{k-1})$ as suggested by Theorem~\ref{thm:mainlist}). Indeed, we simply simply ``forget'' the vertices coloured with colour $d+3,\ldots,k$ in $\alpha$. Let $H$ be the graph without these vertices. The graph $H$ is $d$-degenerate and by Theorem~\ref{thm:mainlist}, we can transform $\alpha|_H$ into $\gamma|_H$ within $\mathcal{O}(n^{d+1})$ steps using colours in $1,\ldots,d+2$. We finally recolour the vertices of $G \setminus H$ one by one with their colours in $\gamma$ to obtain the colouring~$\gamma$. 
\end{proof}

The rest of this section is devoted to prove Theorem~\ref{thm:mainlist}. In order to do it, we need to generalise the notion of full colour to the list-colouring setting. We also need to generalise it to sets of colours, to handle the case where $a > 1$. Given a colouring $\alpha$ of $G$, the set of colour $S$ is \emph{full} if for every vertex $v$ and every colour $c \in S$ one of the following holds:
\begin{enumerate}[label= {(\roman*)}]
	\item $\alpha(v) \in S$,
	\item $v$ has at least one out-neighbour coloured $c$, 
	\item or $c \not \in L(v)$.
\end{enumerate}
We have a property similar as previously: starting from an $a$-feasible list assignment with a colouring~$\alpha$, if $S$ is full then by removing all the vertices $v$ with a colour $\alpha(v) \in S$ from the graph, and removing all the colours from $S$ from all the lists, the resulting assignment is still $a$-feasible. Additionally, the total number of colours has decreased by $|S|$ (but the degeneracy of the graph might not have decreased as much if we had $k$ much larger than  $d +2$). 

In the following, we will denote by $f_a(n, k)$ the maximum diameter of $\cG(G, L)$ over all 
graphs~$G$ with $n$ vertices, and all $a$-feasible assignments $L$ with total number of colours~$k$.

The proof of Theorem~\ref{thm:mainlist} is by induction on the total number of colours $k$. The base case, which is the first point of Theorem~\ref{thm:mainlist}, is the following lemma. This lemma ensures that if the number of excess colours is sufficient (at least half the number of colours), then the diameter is linear\footnote{The proof is similar to the linear diameter obtained in~\cite{BousquetP16} for colourings but adapted to list colourings.}.

\begin{lemma}
\label{lem:linear}
Assume that $k \leq 2a$, then $f_a(n,k) \leq kn$.
\end{lemma}
\begin{proof}
We show by induction on $n$ that if $k \leq 2a+1$, then for any $a$-feasible list assignment $L$ on a graph on $n$ vertices, and any two $L$-colourings $\alpha, \beta$, there is a transformation from $\alpha$ to $\beta$ such that every vertex is recoloured at most $k$ times.

The result is clearly true when $n=1$, since in this case the unique vertex can be recoloured only once. Assume that the result holds for $n-1$. Let $G$ be a graph on $n$ vertices with a degeneracy ordering $v_1, \ldots, v_n$, and an $a$-feasible list assignment $L$ using a total of $k$ colours. Let $\alpha$ and $\beta$ be two $L$-colourings, $H$ be the subgraph obtained after removing $v_n$, and define $d = |d^+(v_1)|$ the number of neighbours of $v_1$.

Using induction on $H$, there is a transformation $\mathcal{S}$ from $\alpha|_H$ to $\beta|_H$ such that every vertex is recoloured at most $k$ times. Note that by assumption, the number of colours available to $v_1$ is at least $d + a + 1$, hence the total number of colours $k$ satisfies, $d + a + 1 \leq k \leq 2a$, and in particular $a \geq d + 1$, and $k \geq d+a+1 \geq 2d+2$. Every time one of the neighbours of $v_1$ is recoloured in the sequence $\mathcal{S}$, we may have to recolour $v_1$ beforehand, so that the colouring remains proper. This happens if the neighbour wants to be recoloured with the current colour of $v_1$.

Every time we have to recolour $v_1$, we have the choice among $a \geq d + 1$ colours $C$ for the new colour of $v_1$. We can look ahead in $\mathcal{S}$ to know which are the next $d+1$ modifications of colours of neighbours of $v_1$ in $\mathcal{S}$. One colour $c$ of $C$ does not appear in these modifications since $|C| \ge d+1$ and the first modified colour is not in $C$ (since we need to recolour $v_1$ at the first step, and then the target colour is the current colour of $v_1$ which is not in $C$). We recolour $v_1$ with $c$. This way, we only need to recolour $v_1$ once out of every $d+1$ times its neighbours are recoloured. Finally, we may need to recolour $v_1$ one last time after the end of $\mathcal{S}$ to colour $v_1$ with its target colour. Since by induction, the neighbours of $v_1$ are recoloured at most $k$ times, the total number of times $v_1$ is recoloured is at most: 

$$ \left\lceil \frac{d k}{d+1} \right\rceil + 1 \leq \frac {d} {d+1} k + 2 \leq k \;,$$

where in the last inequality, we have used the fact that $2 \leq \frac{k}{d+1}$ since $k \geq 2d +2$. This concludes the induction step and proves the result.
\end{proof}

In the induction step of our proof, we will build a set of full colours. For a colouring $\alpha$ and a set $X$ of vertices, we denote by $\alpha(X)= \bigcup_{x \in X} \alpha(x)$. Before stating the main lemmas, let us make the following remark:

\begin{remark}\label{rem:HandG}
Let $G$ be a graph, $L$ be a list assignment and $\alpha$ be a $L$-list colouring. Let $H$ be an induced subgraph of $G$ with list assignment $L'(v)=L(v) \setminus \alpha(N(v) \setminus V(H))$. Then any recolouring sequence $\mathcal{S}$ from $\alpha|_H$ to some colouring $\beta|_H$ also is a (valid) recolouring sequence from $\alpha$ to $\beta$ where $\beta(v)=\beta|_H(v)$ if $v \in H$ and $\alpha(v)$ otherwise.
\end{remark}
By abuse of notation and when no confusion is possible, we will then call $\mathcal{S}$ both the recolouring sequence in $H$ and in $G$.

The following lemma states that, if we already have a set of full colours, then changing it to an other given set can be done without too many additional recolouring steps.

\begin{lemma}
\label{lem:changeFull}
Let $\alpha$ be a colouring of $G$, and $S$ a set of full colours for $\alpha$ with $|S| = a$. For any $S'$ with $|S'| = a$, there exists a colouring $\alpha'$ such that $S'$ is full for $\alpha'$, and there is a transformation from $\alpha$ to $\alpha'$ of length at most $f_a(n, k-a) + (2a + 2)n$.
\end{lemma}
\begin{proof}
    Let $S$ be a set of colours with $|S| = a$ which is full for some $L$-colouring $\alpha$. Let $S'$ be any set of colours of size $a$. The main part of the proof consists in transforming $\alpha$ into a colouring that does not use at all any colour of $S'$ (such a colouring exists since the list assignment is $a$-feasible). 
    
    Let $H$ be the subgraph induced by the vertices not coloured $S$ in $\alpha$, and let $L_H$ be the list assignment for $H$ obtained from $L$ by removing $S$ from all the lists. Since $S$ is full in $\alpha$, $L_H$ is $a$-feasible for $H$.
    
    Consider the following \emph{preference ordering} on the colours: an arbitrary ordering of $[k] \setminus (S \cup S')$, followed by an ordering of $S' \setminus S$, and finally the colours from $S$ last. Let $\gamma$ be the $L$-colouring of $G$ obtained by colouring $G$ greedily from $v_n$ to $v_1$ with this preference ordering. Since $L$ is $a$-feasible, and $|S| = a$, no vertex is coloured with a colour in $S$ in $\gamma$. Indeed $|L(v)| \ge |d^+(v)|+a+1$ and only $|d^+(v)|$ neighbours of $v$ have been coloured when $v$ is coloured. Since in $\gamma$ no vertex has a colour in $S$, $\gamma|_{H}$ is an $L_{H}$-colouring of $H$. By induction hypothesis, there is a recolouring sequence that transforms the colouring $\alpha|_H$ of $H$ into $\gamma|_H$ within at most $f_a(n, k-a)$ steps. By Remark~\ref{rem:HandG}, this recolouring sequence also is a recolouring sequence in $G$.
    We can then recolour the vertices of $G \setminus H$ to their target colour in $\gamma$ in an additional $n$ steps. No conflict can happen at this step since~$\gamma$ is a proper colouring of $G$. 
	
	One can easily check that, in $\gamma$, the set of colours $[k] \setminus (S \cup S')$ is full. Let $K$ be the subgraph of $G$ induced by all the vertices coloured $S'$ in $\gamma$, and let $L_K$ be a list assignment of these vertices where all the colours not in $S \cup S'$ were removed. Then since $[k] \setminus (S \cup S')$ is full, $L_K$ is $a$-feasible. We will recolour $K$ such that no vertex is coloured $S'$ (such a colouring exists because $L_K$ is $a$-feasible, and $|S'| = a$). Since the total number of colours used in $L_K$ is $|S \cup S'| \leq |S| + |S'| = 2a$, this recolouring can be done in at most $f_a(n, 2a) = 2an$ steps by Lemma~\ref{lem:linear}.  By Remark~\ref{rem:HandG}, this recolouring sequence also is a recolouring sequence in $G$. The colouring $\gamma'$ of $G$ that we obtain is such that no vertex is coloured with $c \in S'$. We can finally recolour the vertices of $G$ one by one, starting from $v_n$, choosing a colour of $S'$ if it is available, or leaving it with its current colour otherwise. 
	
	Let $\alpha'$ be the resulting colouring. By construction $\alpha'$ is full for $S'$. The total number of steps to reach $\alpha'$ is at most $f_a(n, k-a) + n +2an + n = f_a(n, k-a) + (2a + 2)n$.
\end{proof}

Using Lemma~\ref{lem:changeFull}, we show that we can incrementally construct a set of full colours.

\begin{lemma}
\label{lem:findFull}
Assume that $k \geq 2a$. For any colouring $\alpha$, there exists a colouring $\beta$ containing a set of full colours $S$ with $|S| = a$, and there is a transformation from $\alpha$ to $\beta$ of length at most $\frac{n}{a} f_a(n, k-a) + 4n^2$.
\end{lemma}
\begin{proof}
    Let $v_1, \ldots, v_n$ be a degeneracy ordering of $G$. A colouring is \emph{full up to step~$i$} for a set of colour $S$ if $|S| = a$, and all the vertices $v_j$ with $j \leq i$ satisfy the conditions $(i), (ii), (iii)$ for the set~$S$. If a colouring $\gamma$ is full up to step~$n$ for the set $S$, then the set $S$ is full.
    
    Note that for any colouring $\alpha$, any set of colours containing $\alpha(v_1),\ldots, \alpha(v_a)$ is full up to step~$a$.
    So we only need to show that given a colouring $\alpha$ which is full up to step~$i$ for some set $S$, we can reach a colouring $\alpha'$ full up to step~$i+a$ for some (potentially different) set $S'$ in at most $f_a(n, k-a) + 4an$ steps. 
	Suppose now that $\alpha$ is full up to step~$i$ but not $i+1$ for some set $S$. Let $S'$ be the colours of the vertices $v_{i+1},\ldots,v_{i+a}$. Up to adding arbitrary colours to $S'$, we can assume that $|S'| = a$. We will then recolour the graph in order to obtain a colouring where $S'$ is full up to step~$i+a$.
	
	Let $H$ be the graph induced by the vertices $v_1, \ldots, v_i$, and $L_H$ be the list assignment of the vertices of $H$ obtained from $L$ by fixing the colours of the vertices outside $H$. In other words, for every vertex $v \in V(H)$, we remove from $L(v)$ the colours of all the vertices of $N(v) \cap \{ v_{i+1},\ldots,v_n \}$. Note that $L_H$ is an $a$-feasible assignment of $H$ since $L$ was an $a$-feasible assignment of $G$. Additionally, $S$ is a set of full colours for the colouring $\alpha|_H$. 
	
	By Lemma~\ref{lem:changeFull}, there is an $L_H$-colouring $\alpha_H'$ of $H$ which is full for $S'$ such that we can transform $\alpha|_H$ into  $\alpha|_H'$ in at most $f_a(n, k-a) + (2a+2)n$ steps. Let $\alpha'$ be the colouring which agrees with~$\alpha$ on the vertices with index larger than $i$, and agrees with $\alpha|_H'$ on $H$. By Remark~\ref{rem:HandG}, $\alpha'$ can be obtained from $\alpha$ into at most $f_a(n, k-a) + (2a+2)n$ steps. By construction, $S'$ is full up to step~$i$, and since the vertices $v_{i+1}, \ldots v_{i+a}$ are coloured with colours in $S'$, it is full up to step~$i+a$.
	
	Finally, this procedure needs to be repeated at most $\left\lceil \frac {n}{a} \right\rceil - 1 \leq \frac n a$ times (the minus one coming from the fact that at the beginning we had for free a set of colours full for $v_1,\ldots,v_a$). After this many steps, we obtain a colouring full up to step~$n$ for some set $S$ with $|S|=a$, which concludes the proof.
\end{proof}

Finally, we can use the two previous lemma to get the following recursive inequality.

\begin{lemma}
\label{lem:recursionStep}
Let $k \geq 2a$, then $f_a(n,k) \leq (\frac{2n}{a} + 3) f_a(n, k-a) + 10 \cdot n^2$.
\end{lemma}
\begin{proof}
Let $L$ be an $a$-feasible list assignment of $G$, and $\alpha$ and $\beta$ be two $L$-colourings of $G$. By Lemma~\ref{lem:findFull}, there exists a colouring $\alpha'$ and a set of colours $S_\alpha$ with $|S_\alpha| = a$ which is full for $\alpha'$ such that the colouring $\alpha'$ can be reached from $\alpha$ in at most $\frac n a f_a(n, k-a) + 4n^2$ steps. 
Similarly, there exists  a colouring $\beta'$ and a set of colours $S_\beta$ with $|S_\beta |=a$ such that $S_\beta$ is full for $\beta'$  such that the colouring $\beta'$ can be reached from $\beta$ in at most $\frac n a f_a(n, k-a) + 4n^2$ steps.  By Lemma~\ref{lem:changeFull}, using an additional $f_a(n, k-a) + 4n$ steps, we can get a colouring $\beta''$ from $\beta'$ such that the set of full colours in $\beta''$ and $\alpha'$ is the same (namely $S_\alpha$). 

Let $S$ be some set of colours disjoint from $S_\alpha$ with $|S| = a$. Let $\gamma$ be a colouring of $G$ that does not use any colour of $S_\alpha$ (such a colouring exists since the list assignment of $G$ is $a$-feasible).

Let $G_\alpha$ be the graph $G$ where the vertices coloured with colours in $S_\alpha$ have been deleted and the colours in $S_\alpha$ removed from the list assignment. Since $S_\alpha$ is full for $\alpha'$, the list assignment of $G_\alpha$ is $a$-feasible. Note that $\gamma|_{G_\alpha}$ is a proper colouring of $G_\alpha$.
So by induction, it is possible to recolour $\alpha'|_{G_\alpha}$ into $\gamma|_{G_\alpha}$ in at most $f_a(n, k-a)$ steps. Since the vertices of $W:=V(G) \setminus V(G_\alpha)$ are coloured with colours in $S_\alpha$, and since~$\gamma$ does not use any of these colours, one can finally recolour the vertices of $W$ one by one from their colours in $S_\alpha$ to their target colours in $\gamma$.

Let $G_{\beta''}$ be the graph where the vertices coloured with $S_{\beta''}=S_\alpha$ in $\beta''$ have been deleted. One can similarly recolour $\beta''|_{G_{\beta''}}$ into $\gamma|_{G_{\beta''}}$ in at most $f_a(n, k-a)$ steps. Since vertices of $W':=V(G) \setminus V(G_{\beta''})$ are coloured with colours in $S_\alpha$, we can also can recolour the vertices of $W'$ one by one from their colours in $S_\alpha$ to their target colours in $\gamma$.

The total number of steps to transform $\alpha$ into $\beta$ is at most

\begin{align*} 
f(n,k) \le &\quad 2\left(\frac n a f_a(n, k-a) + 4n^2\right) + (f_a(n, k-a) + 4n) + 2 f_a(n, k-a) + 2n \\
\le& \left(\frac {2n} {a} + 3 \right) f_a(n, k-a) + 10n^2
\end{align*}

%
%


\end{proof}

\begin{lemma}\label{lem:formgen}
    For all $k \geq 1$ and $a\ge 1$, we have $f_a(n,k) \leq Ck n \left(\frac {2n} {a} + 3\right)^{\left\lceil \frac {k}{a}\right\rceil - 2} $ for some constant $C > 0$.
\end{lemma}
\begin{proof}
    We prove it by induction on the total number of colours $k$. The base case is when $k \le 2a$ and simply is a consequence of Lemma~\ref{lem:linear} (note that $\left\lceil \frac ka \right\rceil = 2$ since $k > a$). 
     The induction step is obtained using Lemma~\ref{lem:recursionStep}.
\end{proof}

We now have all the ingredients to prove Theorem~\ref{thm:mainlist}, which completes the proof of Theorem~\ref{thm:main}.

\begin{proof}[Proof of Theorem~\ref{thm:mainlist}]
The first point is the result from Lemma~\ref{lem:linear}. The second and last points are consequences of Lemma~\ref{lem:formgen} with the corresponding values of $a$. Let us prove the third point. Since $k \le (1 + \frac 1\varepsilon )a$, we have $\left\lceil \frac ka \right\rceil -2 \le \lceil 1 + \frac 1 \varepsilon\rceil -2 \le \lceil \frac 1\varepsilon \rceil - 1$. Consequently, the function given by Lemma~\ref{lem:formgen} is $O(n^{\lceil 1/\varepsilon \rceil})$. Note that for the second and third point, the constant does not depend on $k$ or $a$ since $k/a$ is a constant (but it depends on~$\varepsilon$ in the third point). 
\end{proof}

\section{Planar bipartite graphs}\label{sec:bip}

The rest of this article is devoted to prove Cereceda's quadratic conjecture for planar bipartite graphs.

\bipPlanar*


We start by defining the \emph{level} of a vertex. Given a graph $G$ all the vertices with degree at most~$3$ have level~$1$. For any $i > 1$, the vertices with level~$i$ are the ones with degree at most~$3$ after removal of all the vertices of level~$j < i$. 

Let $G$ be a planar graph given with its representation. The \emph{size} of a face is the number of vertices incident to the face. The proof of Theorem~\ref{thm:planarBip} is based on a discharging proof. We are assuming that the results does not hold. The proof then goes in two steps. First, we show that any planar bipartite graph must contain at least one "configuration" of a list of (here two) configurations.
Then, we prove that any counter-example cannot contain any other these configurations, which implies that this counter-example cannot exist. Let us first give the two configurations we will need.

\begin{lemma}
\label{lem:struct}
Let $G$ be a planar bipartite graph. At least one of the following holds:
\begin{enumerate}[label={(\roman*)}]
\item $G$ contains a vertex of degree at most $2$ or, 
\item $G$ contains a vertex of degree $3$ incident to three vertices of level at most $2$, and incident to a face of size four.
\end{enumerate}
\end{lemma}
\begin{proof}
We use a discharging argument to prove the result. Suppose by contradiction that there exists a planar bipartite graph $G$ satisfying none of the conditions above. Assign to each vertex $v$ the weight $\deg(v) -4$, and to each face $f$ the weight $\deg(f) -4$. The total weight assigned this way is:
\begin{align*}
\sum_{v \in G} (\deg(v) -4) + \sum_{f, \text{face}} (\deg(f) -4) = 2|E| - 4n + 2|E| - 4|F| = 4(|E| - n - |F|) = -8
\end{align*}
where $n, |E|$ and $|F|$ are the number of vertices, edges, and faces of $G$. The last equality comes from Euler's formula for planar graphs.

Now the goal is to re-allocate the weights on the faces and the vertices to obtain a total weight that is non-negative, a contradiction with the equation above. Note that all the faces as well as vertices of degree at least~$4$ have non-negative weights. In order to get a contradiction, we apply the following procedure and prove that after applying it all the vertices and faces have non-negative weights:
\begin{itemize}
\item Every vertex of level at least $3$ gives weight $1$ to all its neighbours of degree $3$.
\item Every face of size at least $6$ gives weight $\frac 1 3$ to all its incident vertices. 
\end{itemize}
After applying this transformation, every face has a non-negative weight. Indeed, the faces of lengt $4$ still have weight $0$, and the faces of length at least $6$ gave a total of $\frac {\deg(f)} 3 \leq \deg(f) - 4$ since $\deg(f) \geq 6$.

Vertices of level $2$ initially have non-negative weights since their degree is at least four, and they did not give away any of their weights. Hence their weights remain non-negative. Vertices of level at least $3$ must have, by definition at least $4$ neighbours of degree larger than $3$ (otherwise they would be of level $2$), and thus keep a positive weight. 

Finally, let $v$ be a vertex of degree $3$. If $v$ is adjacent to a vertex of level at least $3$, it receives weight $1$ during the first step, and its weight is non-negative. 
Otherwise, $v$ is adjacent to three vertices of level at most $2$, and then, by assumption on $G$, the three faces incident to $v$ must have size at least $6$. Then, $v$ received $\frac 1 3$ from each of the faces during the second step, and its weight is non-negative. 

Hence, after applying the the procedure, every face and every vertex has a non-negative weight, a contradiction with the fact that the total weight is negative.
\end{proof}

We can now prove the Theorem~\ref{thm:planarBip}. The proceeds by iteratively reducing the size of the graph, either by removing a vertex, or contracting two vertices together.

\begin{proof}[Proof of Theorem~\ref{thm:planarBip}]
We will show by induction on the size of the graph $G$ that for any two $5$ colourings of $G$, $\alpha$ and $\beta$, there exists a transformation from $\alpha$ to $\beta$ such that each vertex is recoloured at most $cn$ times, for some constant $c$ defined later. 

If $G$ is reduced to a single vertex, the result is immediate.
Otherwise, let $n > 1$. One of the two conditions of Lemma~\ref{lem:struct} must occur.
\smallskip

\noindent \textbf{Case (i).} The graph $G$ contains a vertex $v$ of degree at most $2$.

Let $H$ be the subgraph obtained after removing $v$. By induction, there exists a transformation $\mathcal{S}$ from $\alpha|_H$ to $\beta|_H$ such that each vertex is recoloured at most $c(n-1)$ times. Given this transformation, we are going to produce a transformation $\mathcal{S}'$ for the whole graph $G$. We follow the transformation steps of $\mathcal{S}$. If the same recolouring step is possible in $G$ we do it. Otherwise, it means that one of the (at most) two neighbours of $v$ is recoloured, and that the target colour is the current colour $c$ of $v$. So we need to recolour $v$ to make the transition possible. To recolour $v$, we have at least two possible new colour choices $X$ for $v$ distinct from $c$. We can look forward at the sequence $\mathcal{S}$ to see which colour $c'$ will appear next in the neighbourhood of $v$. We choose for $v$ a colour of $X$ distinct from $c'$. In particular, the choice of $c'$ ensures that $v$ will not be recoloured next time a neighbour of $v$ is recoloured. At the end of the transformation $\mathcal{S}$, we may need one additional step to recolour $v$ to its target colour. 

Thus in total, $v$ is recoloured at most $\frac {2 c(n-1)}{2} + 1 \leq cn$ times (which holds as long as $c \ge 1$).
\smallskip

\noindent
\textbf{Case (ii).} The graph $G$ contains a vertex $v$ of degree three such that:
\begin{itemize}
\item all three neighbours of $v$ have level at most $2$, 
\item $v$ is incident to a face of length $4$.
\end{itemize}
Let $f$ be the face of length $4$ incident to $v$, and let $w$ be the vertex opposite to $v$ on $f$. Let $u$ be the vertex adjacent to $v$, and not incident to $f$. Let $v_1$ and $v_2$ be the two remaining vertices adjacent to $f$ (see Figure~\ref{fig:caseii}). In this case, we want to reduce the graph by merging the two vertices~$v$ and~$w$ (maintaining a bipartite planar graph). For this we will need the following result.

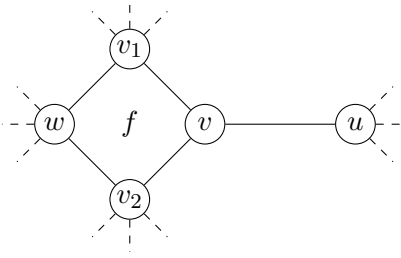
\begin{figure}[!ht]
\centering
\begin{tikzpicture}[nstyle/.style={
	draw, circle, inner sep = 0pt, minimum size=15pt
	}]
    \node[nstyle] (v) at (0,0) {$v$} ;
    \node[nstyle] (u) at (2,0) {$u$} ;
    \node[nstyle] (w) at (-2,0) {$w$} ;
    \node[nstyle] (v1) at (-1,1) {$v_1$} ;
    \node[nstyle] (v2) at (-1,-1) {$v_2$} ;
    \node at (-1,0) {$f$} ;
    \draw (u) -- (v) (v) -- (v1) (v) -- (v2) (v2) -- (w) (v1) -- (w) ;  
    \draw[dashed] (u) -- +(0.5,0.5) (u) -- +(0.5, -0.5) (u) -- +(0.7, 0) ;
    \draw[dashed] (v1) -- +(0.5,0.5) (v1) -- +(-0.5, 0.5) (v1) -- +(0, 0.7) ;
    \draw[dashed] (v2) -- +(0.5,-0.5) (v2) -- +(-0.5, -0.5) (v2) -- +(0, -0.7) ;
    \draw[dashed] (w) -- +(-0.5,0.5) (w) -- +(-0.5, -0.5) (w) -- +(-0.7, 0) ;
    
\end{tikzpicture}
\caption{Reducible configuration in case $(ii)$.}
\label{fig:caseii}
\end{figure}

\begin{claim}\label{clm:bip}
Starting from $\alpha$, we can reach a colouring $\alpha'$ such that $\alpha'(v) = \alpha'(w)$. This colouring can be reached by recolouring each vertex at most twice. 
\end{claim}
\begin{proof}
Since $v_1$ and $v_2$ are adjacent to $w$, their colour is different from $\alpha(w)$. Hence, the only obstruction to recolour directly $v$ with $\alpha(w)$ is if $u$ is coloured with $\alpha(w)$. Since $u$ has level at most~$2$, it has at most~$3$ neighbours with degree larger than~$3$. Let $c$ be a colour different from these three neighbours, and different from $\alpha(w)$. First, we recolour all the degree three neighbours of~$u$ which are coloured~$c$ with any other colour (note that~$w$ is not recoloured since $\alpha(w) \ne c$, but~$v$ might be recoloured). Then, we recolour~$u$ with~$c$, and finally~$v$ can be recoloured with~$\alpha(w)$ (since $G$ is bipartite, the recolouring sequence ensures that the only neighbor of $v$ that is recoloured is $u$). Every vertex is recoloured at most once, except~$v$ which may be recoloured twice.

\end{proof}

By Claim~\ref{clm:bip} starting from $\alpha$ and $\beta$, we obtain two colourings $\alpha'$ and $\beta'$ such that $\alpha'(v) = \alpha'(w)$ and $\beta'(v) = \beta'(w)$. Let $H$ be the graph obtained after merging~$v$ and~$w$ into a single vertex $x$. Clearly $H$ is still planar, and it is still bipartite since $v$ and~$w$ are on the same side of the bipartition of $G$. Moreover, $\alpha'$ and $\beta'$ are proper colourings of $H$. By induction, there is a transformation from $\alpha'$ to $\beta'$ in $H$ such that each vertex is recoloured at most $c(n-1)$ times. By applying each transformation on $x$ to both $v$ and $w$, this gives a transformation in $G$ from $\alpha'$ to $\beta'$ such that each vertex is recoloured at most $c(n-1)$ times. 

This gives a transformation from $\alpha$ to $\beta$ where each vertex is recoloured at most $c(n-1) + 4 \leq cn$ times, taking $c = 4$.
\end{proof}

\bibliography{biblio.bib}

\bibliographystyle{abbrv}

\end{document}